\documentclass[11pt]{article}
\usepackage[letterpaper,margin=1in]{geometry}

\usepackage{hyperref}
\usepackage{amsmath,amssymb,amsfonts} 
\hypersetup{
pdfborder = 0 0 0,
colorlinks=true,
citecolor=blue,
linkcolor=blue,
urlcolor=blue,
}
\usepackage{graphicx}

\usepackage{amsthm}
\newtheorem{theorem}{Theorem}[section]
\newtheorem{lemma}[theorem]{Lemma}
\newtheorem{corollary}[theorem]{Corollary}
\newtheorem{definition}[theorem]{Definition}

\newcommand{\seclabel}[1]{\label{sec:#1}}
\newcommand{\secref}[1]{Section~\ref{sec:#1}}

\newcommand{\lemref}[1]{Lemma~\ref{lem:#1}}

\newcommand{\thmref}[1]{Theorem~\ref{thm:#1}}

\let\footnotesize\small

\newcommand{\card}[1]{\left|#1\right|}
\newcommand{\set}[1]{\left\{#1\right\}}
\newcommand{\ceil}[1]{\left\lceil#1\right\rceil}

\newcommand{\E}{\textbf{E}}

\allowdisplaybreaks
\begin{document}

\begin{titlepage}
  \title{Smoothed Analysis of Dynamic Networks}

  \author{
    Michael Dinitz\thanks{Supported in part by NSF grant \#1464239}\\
    Johns Hopkins University\\
    \texttt{mdinitz@cs.jhu.edu}
    \and
    Jeremy T. Fineman\thanks{Supported in part by NSF
          grants CCF-1218188 and CCF-1314633}\\
    Georgetown University\\
    \texttt{jfineman@cs.georgetown.edu}
    \and
    Seth Gilbert\\
    National University of Singapore\\
    \texttt{seth.gilbert@comp.nus.edu.sg}
    \and
    Calvin Newport\\
    Georgetown University\\
    \texttt{cnewport@cs.georgetown.edu}
  }  
  \date{}

  \maketitle

  \begin{abstract}
    We generalize the technique of {\em smoothed analysis} to
    distributed algorithms in dynamic network models.  Whereas
    standard smoothed analysis studies the impact of small random
    perturbations of input values on algorithm performance metrics,
    dynamic graph smoothed analysis studies the impact of random
    perturbations of the underlying changing network graph topologies.
    Similar to the original application of smoothed analysis, our goal
    is to study whether known strong lower bounds in dynamic network
    models are \emph{robust} or \emph{fragile}: do they withstand
    small (random) perturbations, or do such deviations push the
    graphs far enough from a precise pathological instance to enable
    much better performance?  Fragile lower bounds are likely not
    relevant for real-world deployment, while robust lower bounds
    represent a true difficulty caused by dynamic behavior.  We apply
    this technique to three standard dynamic network problems with
    known strong worst-case lower bounds: random walks, flooding, and
    aggregation.  We prove that these bounds provide a spectrum of
    robustness when subjected to smoothing---some are extremely
    fragile (random walks), some are moderately fragile / robust
    (flooding), and some are extremely robust (aggregation).
   \end{abstract}

\setcounter{page}{0}
\thispagestyle{empty}
\end{titlepage}

\section{Introduction}\seclabel{intro}

Dynamic network models describe networks with topologies that change
over time (c.f.,~\cite{dynamic-overview}).  They are used to capture
the unpredictable link behavior that characterize challenging
networking scenarios; e.g., connecting and coordinating moving
vehicles, nearby smartphones, or nodes in a widespread and fragile
overlay.  Because fine-grained descriptions of link behavior in such
networks are hard to specify, most analyses of dynamic networks rely
instead on a worst-case selection of graph changes.
This property is crucial to the usefulness of these analyses, as it
helps ensure the results persist in real deployment.

A problem with this worst case perspective is that it often leads to
extremely strong lower bounds.  These strong results motivate a key
question: {\em Is this bound {\em robust} in the sense that it
  captures a fundamental difficulty introduced by dynamism, or is the
  bound {\em fragile} in the sense that the poor performance it
  describes depends on an exact sequence of adversarial changes?}
Fragile lower bounds leave open the possibility of algorithms that
might still perform well in practice. By separating fragile from
robust results in these models, therefore, we can expand the
algorithmic tools available to those seeking useful guarantees in
these challenging environments.
 
In the study of traditional algorithms, an important technique for
explaining why algorithms work well in practice, despite disappointing
worst case performance, is {\em smoothed
  analysis}~\cite{SpielmanT04,SpielmanT09}.  This approach studies the
expected performance of an algorithm when the inputs are slightly
randomly perturbed. If a strong lower bound dissipates after a small
amount of smoothing, it is considered fragile---as it depends on a
carefully constructed degenerate case.  Note that this is different
from an ``average-case" analysis, which looks at instances drawn from
some distribution.  In a smoothed analysis, you still begin with an
adversarially chosen input, but then slightly perturb this choice.  Of
course, as the perturbation grows larger, the input converges toward
random. (Indeed, in the original smoothed analysis
papers~\cite{SpielmanT04,SpielmanT09}, the technique is described as
interpolating between worst and average case analysis.)

In this paper, we take the natural next step of adapting smoothed
analysis to the study of distributed algorithms in dynamic networks.
Whereas in the traditional setting smoothing typically perturbs
numerical input values, in our setting we define smoothing to perturb
the network graph through the random addition and deletion of edges.
We claim that a lower bound for a dynamic network model that improves
with just a small amount of graph smoothing of this type is fragile,
as it depends on the topology evolving in an exact manner.  On the
other hand, a lower bound that persists even after substantial
smoothing is robust, as this reveals a large number of similar graphs
for which the bound holds.

\paragraph{Results.}
We begin by providing a general definition of a dynamic network model
that captures many of the existing models already studied in the
distributed algorithms literature.  At the core of a dynamic network
model is a dynamic graph that describes the evolving network topology.
We provide a natural definition of {\em smoothing} for a dynamic graph
that is parameterized with a {\em smoothing factor}
$k\in \{0,1,...,\binom{n}{2}\}$.  In more detail, to $k$-smooth a
dynamic graph ${\cal H}$ is to replace each static graph $G$ in
${\cal H}$ with a smoothed graph $G'$ sampled uniformly from the space
of graphs that are:
(1) within edit distance\footnote{The notion of edit distance we
use in this paper is the number of edge additions/deletions needed to transform
one graph to another, assuming they share the same node set.} $k$ of $G$, and 
(2) are allowed by the relevant dynamic network model (e.g., if the
model requires the graph to be connected in every round, smoothing
cannot generate a disconnected graph).

We must next argue that these definitions allow for useful discernment
between different dynamic network lower bounds.  To this end, we use
as case studies three well-known problems with strong lower bounds in
dynamic network models: flooding, random walks, and aggregation.  For
each problem, we explore the robustness/fragility of the existing
bound by studying how it improves under increasing amounts of
smoothing.  Our results are summarized in Table~\ref{fig:results}.  We
emphasize the surprising variety in outcomes: these results capture a
wide spectrum of possible responses to smoothing, from quite fragile
to quite robust.

For the minimal amount of smoothing ($k=1$), for example, the
$\Omega(2^n)$ lower bound for the hitting time of a random walk in
connected dynamic networks (established in~\cite{avin:2008}) decreases
by an {\em exponential} factor to $O(n^3)$, the $\Omega(n)$ lower
bound for flooding time in these same networks (well-known in
folklore) decreases by a {\em polynomial} factor to
$O(n^{2/3}\log{n})$, and the $\Omega(n)$ lower bound on achievable
competitive ratio for token aggregation in pairing dynamic graphs
(established in~\cite{cornejo:2012}) decreases by only a {\em
  constant} factor.
 
As we increase the smoothing factor $k$, our upper bound on random
walk hitting time decreases as $O(n^3/k)$, while our flooding upper
bound reduces slower as $O(n^{2/3}\log{n}/k^{1/3})$, and our
aggregation bound remains in $\Omega(n)$ for $k$ values as large as
$\Theta(n/\log^2{n})$.  In all three cases we also prove tight or near
tight lower bounds for all studied values of $k$, showing that these
analyses are correctly capturing the impact of smoothing.
 
Among other insights, these results indicate that the exponential
hitting time lower bound for dynamic walks is extremely fragile, while
the impossibility of obtaining a good competitive ratio for dynamic
aggregation is quite robust.  Flooding provides an interesting
intermediate case.  While it is clear that an $\Omega(n)$ bound is
fragile, the claim that flooding can take a polynomial amount of time
(say, in the range $n^{1/3}$ to $n^{2/3}$) seems well-supported.

\begin{table}
\footnotesize
\begin{tabular}{|lllll|}  
\hline
& Graph &$k$-Smoothed Algorithm & $k$-Smoothed Lower Bound & $0$-Smoothed Lower Bound \\ \hline
Flooding & Connected &$O(n^{2/3}\log{n}/k^{1/3})$ & $\Omega(n^{2/3}/k^{1/3})$ & $\Omega(n)$ \\
Hitting Time & Connected  &  $O(n^3/k)$ &  $\Omega(n^{5/2}/(\sqrt{k}\log{n})$ & $\Omega(2^n)$ \\
Aggregation & Paired & $O(n)$-competitive & $\Omega(n)$-competitive & $\Omega(n)$-competitive \\
\hline
\end{tabular}
\caption{A summary of our main results. The columns labelled ``$k$-smoothed" assume $k >0$. Different
results assume different upper bounds on $k$.}
\label{fig:results}
\end{table}

\paragraph{Next Steps.}
The definitions and results that follow represent a first (but far
from final) step toward the goal of adapting smoothed analysis to the
dynamic network setting.  There are many additional interesting
dynamic network bounds that could be subjected to a smoothed analysis.
Moreover, there are many other reasonable definitions of smoothing
beyond the ones we use here.  While our definition is natural and our
results suggestive, for other problems or model variations other
definitions might be more appropriate.  Rather than claiming that our
approach here is the ``right" way to study the fragility of dynamic
network lower bounds, we instead claim that smoothed analysis
generally speaking (in all its various possible formulations) is an
important and promising tool when trying to understand the fundamental
limits of distributed behavior in dynamic network settings.
 
\paragraph{Related Work.}
Smoothed analysis was introduced by Spielman and
Teng~\cite{SpielmanT04,SpielmanT09}, who used the technique to explain
why the simplex algorithm works well in practice despite strong
worst-case lower bounds.  It has been widely applied to traditional
algorithm problems (see~\cite{spielman:2009} for a good introduction
and survey).  Recent interest in studying distributed algorithms in
dynamic networks was sparked by Kuhn et al.~\cite{kuhn:2010}.  Many
different problems and dynamic network models have since been
proposed;
e.g.,~\cite{kuhn:2011,haeupler:2011,dutta:2013,clementi:2012,augustine:2012,denysyuk:2014,newport:2014,ghaffari:2013}
(see~\cite{dynamic-overview} for a survey).  The dynamic random walk
lower bound we study was first proposed by Avin
et~al.~\cite{avin:2008}, while the dynamic aggregation lower bound we
study was first proposed by Cornejo et~al.~\cite{cornejo:2012}.  We
note other techniques have been proposed for exploring the fragility
of dynamic network lower bounds.  In recent work, for example,
Denysyuk et al.~\cite{denysyuk:2014} thwart the exponential random
walk lower bound due to~\cite{avin:2008} by requiring the dynamic
graph to include a certain number of static graphs from a well-defined
set, while work by Ghaffari et~al.~\cite{ghaffari:2013} studies the
impact of adversary strength, and Newport~\cite{newport:2014} studies
the impact of graph properties, on lower bounds in the dynamic radio
network model.

\section{Dynamic Graphs, Networks, and Types}

There is no single dynamic network model.  There are, instead, many
different models that share the same basic behavior: nodes executing
synchronous algorithms are connected by a network graph that can
change from round to round.  Details on how the graphs can change and
how communication behaves given a graph differ between model types.

In this section we provide a general definition for a dynamic network
models that captures many existing models in the relevant literature.
This approach allows us in the next section to define smoothing with
sufficient generality that it can apply to these existing models.  We
note that in this paper we constrain our attention to {\em oblivious}
graph behavior (i.e., the changing graph is fixed at the beginning of
the execution), but that the definitions that follow generalize in a
straightforward manner to capture adaptive models (i.e., the changing
graph can adapt to behavior of the algorithm).

\noindent {\bf Dynamic Graphs and Networks.}
Fix some node set $V$, where $n=|V|$.  A {\em dynamic graph}
${\cal H}$, defined with respect to $V$, is a sequence $G_1,G_2,...$,
where each $G_i= (V,E_i)$ is a graph defined over nodes $V$.  If this
is not an infinite sequence, then the \emph{length} of $\mathcal H$ is
$|\mathcal H|$, the number of graphs in the sequence.  A {\em dynamic
  network}, defined with respect to $V$, is a pair, $({\cal H}, C)$,
where ${\cal H}$ is a dynamic graph, and $C$ is a {\em communication
  rules} function that maps {\em transmission patterns} to {\em
  receive patterns}.  That is, the function takes as input a static
graph and an assignment of messages to nodes, and returns an
assignment of received messages to nodes.  For example, in the
classical radio network model $C$ would specify that nodes receive a
message only if exactly one of their neighbors transmits, while in the
$\mathcal{LOCAL}$ model $C$ would specify that all nodes receive all
messages sent by their neighbors.
Finally, an {\em algorithm} maps process definitions to nodes in $V$.

Given a dynamic network $({\cal H}, C)$ and an algorithm ${\cal A}$,
an execution of ${\cal A}$ in $({\cal H}, C)$ proceeds as follows: for
each round $r$, nodes use their process definition according to
${\cal A}$ to determine their transmission behavior, and the resulting
receive behavior is determined by applying $C$ to ${\cal H}[r]$ and
this transmission pattern.

\noindent {\bf Dynamic Network Types.}
When we think of a dynamic network model suitable for running
executions of distributed algorithms, what we really mean is a
combination of a description of how communication works, and a set of
the different dynamic graphs we might encounter.  We formalize this
notion with the concept of the {\em dynamic network type}, which we
define as a pair $({\cal G}, C)$, where ${\cal G}$ is a set of dynamic
graphs and $C$ is a communication rules function.  For each
${\cal H} \in {\cal G}$, we say dynamic network type $({\cal G}, C)$
contains the dynamic network $({\cal H}, C)$.

When proving an upper bound result, we will typically show that the
result holds when our algorithm is executed in any dynamic network
contained within a given type.  When proving a lower bound result, we
will typically show that there exists a dynamic network contained
within the relevant type for which the result holds.
In this paper, we will define and analyze two existing dynamic network types:
{\em ($1$-interval)  connected networks}~\cite{kuhn:2010,kuhn:2011,haeupler:2011,dutta:2013}, in which the graph in each round
is connected and $C$ describes reliable broadcast to neighbors in the graph,
and {\em pairing networks}~\cite{cornejo:2012},
in which the graph in each round is a matching and $C$ describes reliable 
message passing with each node's neighbor (if any).

\section{Smoothing Dynamic Graphs} \seclabel{smoothing}

We now define a version of smoothed analysis that is relevant to
dynamic graphs.  To begin, we define the {\em edit distance} between
two static graphs $G=(V,E)$ and $G'=(V,E')$ to be the minimum number
of edge additions and removals needed to transform $G$ to $G'$.  With
this in mind, for a given $G$ and $k\in \{0,1,...,\binom{n}{2}\}$, we
defined the set:

\[ editdist(G,k) = \{ G' \mid \text{the edit distance between $G$ and $G'$ is no more than $k$}\}.\]

\noindent Finally, for a given set of dynamic graphs ${\cal G}$,
we define the set: 

\[ allowed({\cal G}) = \{ G\mid \exists {\cal H} \in {\cal G} \text{ such that } G\in {\cal H}\}.\]

\noindent In other words, $allowed$ describes all graphs that show up
in the dynamic graphs contained in the set ${\cal G}$.  Our notion of
smoothing is always defined with respect to a dynamic graph set
${\cal G}$.  Formally:

\begin{definition}
  Fix a set of dynamic graphs ${\cal G}$, a dynamic graph
  ${\cal H} \in {\cal G}$, and {\em smoothing factor}
  $k\in \{0,1,...,\binom{n}{2}\}$.  To {\em $k$-smooth} a static graph
  $G\in {\cal H}$ (with respect to ${\cal G}$) is to replace $G$ with
  a graph $G'$ sampled uniformly from the set
  $editdist(G,k) \cap allowed({\cal G})$.  To {\em $k$-smooth} the
  entire dynamic graph ${\cal H}$ (with respect to ${\cal G}$), is to
  replace ${\cal H}$ with the dynamic graph ${\cal H'}$ that results
  when we $k$-smooth each of its static graphs.
\end{definition}

\noindent We will also sometimes say that $G'$ (resp. ${\cal H'}$) is
a {\em $k$-smoothed} version of $G$ (resp. ${\cal H}$), or simply a
{\em $k$-smoothed} $G$ (resp. ${\cal H}$).  We often omit the dynamic
graph set ${\cal G}$ when it is clear in context. (Typically,
${\cal G}$ will be the set contained in a dynamic network type under
consideration.)

\paragraph{Discussion.}
Our notion of $k$-smoothing transforms a graph by randomly adding or
deleting $k$ edges. A key piece of our definition is that smoothing a
graph with respect to a dynamic graph set cannot produce a graph not
found in any members of that set. This restriction is particularly
important for proving lower bounds on smoothed graphs, as we want to
make sure that the lower bound results does not rely on a dynamic
graph that could not otherwise appear. For example, if studying a
process in a dynamic graph that is always connected, we do not want
smoothing to disconnect the graph---an event that might trivialize
some bounds.

\section{Connected and Pairing Dynamic Network Types}\seclabel{types}

This section defines the two dynamic network types studied in this
paper: the {\em connected network
  type}~\cite{kuhn:2010,kuhn:2011,haeupler:2011,dutta:2013}, and the
{\em pairing network type}~\cite{cornejo:2012}.  We study random walks
(\secref{random-walks}) and flooding (\secref{flooding}) in the
context of the connected network type, whereas we study token
aggregation (\secref{aggregation}) in the context of the pairing type.

\subsection{Connected Network}
The {\em connected network
  type}~\cite{kuhn:2010,kuhn:2011,haeupler:2011,dutta:2013} is defined
as $({\cal G}_{conn}, C_{conn})$, where ${\cal G}_{conn}$ contains
every dynamic graph (defined with respect to our fixed node set $V$)
in which every individual graph is connected, and where $C_{conn}$
describes reliable broadcast (i.e., a message sent by $u$ in rounds
$r$ in an execution in graph ${\cal H}$ is received by every neighbor
of $u$ in ${\cal H}[r]$).

\paragraph{Properties of Smoothed Connected Networks.}
For our upper bounds, we show that if certain edges are added to the
graph through smoothing, then the algorithm makes enough progress on
the smoothed graph.  For our lower bounds, we show that if certain
edges are not added to the graph, then the algorithm does not make
much progress.  The following lemmas bound the probabilities that
these edges are added.  The proofs roughly amount to showing that
sampling uniformly from $editdist(G,k) \cap allowed({\cal G}_{conn})$
is similar to sampling from $editdist(G,k)$.

The first two lemmas are applicable when upper-bounding the
performance of an algorithm on a smoothed dynamic graph.  The first
lemma states that the $k$-smoothed version of graph $G$ is fairly
likely to include at least one edge from the set $S$ of helpful edges.
The second lemma, conversely, says that certain critical edges that
already exist in $G$ are very unlikely to be removed in the smoothed
version.

\begin{lemma} \label{lem:connlb}
  There exists constant $c_1>0$ such that the following holds.
  Consider any graph $G \in allowed({\cal G}_{conn})$.  Consider also
  any nonempty set $S$ of potential edges and smoothing value $k \leq n/16$ with
  $k\card{S} \leq n^2/2$.  Then with probability at least
  $c_1k\card{S}/n^2$, the $k$-smoothed graph $G'$ of $G$ contains at
  least one edge from $S$.
\end{lemma}
\begin{proof}
  We start by noting that sampling a graph $G_D$ uniformly from
  $editdist(G,k)$ and then rejecting (and retrying) if
  $G_D \not\in allowed({\cal G}_{conn})$ is equivalent to sampling
  uniformly from $k$-smoothed graphs.  Let $A$ be the event that $G_D$
  includes an edge from $S$.  Let $B$ be the event that $G_D$ is
  connected. Our goal is to lower-bound $\Pr[A|B]$ starting from
  $\Pr[A|B] \geq \Pr[\text{$A$ and $B$}]$.

  The proof consists of two cases, depending on whether the input
  graph $G=(V,E)$ contains an edge from $S$.  For the first case,
  suppose that it does.  Choose an arbitrary edge $e \in S\cap E$, and
  choose an arbitrary spanning tree $T$ or $G$.  Let $B_{Te}$ be the
  event that $G_D$ contains all edges in $T\cup\set{e}$.  Note that if
  $B_{Te}$ occurs, then $G_D$ is both connected and contains an edge
  from $S$.  Thus, $\Pr[\text{$A$ and $B$}] \geq \Pr[B_{Te}]$, and we
  need only bound $\Pr[B_{Te}]$.  Sampling a graph $G_D$ from
  $editdist(G,k)$ is equivalent to selecting up to $k$ random edges
  from among all potential $\binom{n}{2}$ edges and toggling their
  status.  Consider the $i$th edge toggled through this process.  The
  probability that the edge is one of the $n$ edges in $T$ or $e$ is
  at most $2n/\binom{n}{2}$, where the loose factor of $2$ arises from
  the fact that the $i$th edge is only selected from among the
  $\binom{n}{2}-i$ remaining edges (and
  $\binom{n}{2}-i \geq \binom{n}{2}/2$ for
  $k\leq\binom{n}{2}/2$). Taking a union bound over at most $k$ edge
  choices, we have $\Pr[\text{not $B_{Te}$}] \leq 2kn/\binom{n}{2}$.  Thus
  $\Pr[B_{Te}] \geq 1-2kn/\binom{n}{2} \geq 1/2$ for $k \leq n/16$.

  The second case is that $S \cap E = \emptyset$. As before, choose
  any arbitrary spanning tree $T$ in $G$.  Let $B_T$ be the event that
  $T$ is in $G_D$.  Since $B_T \subseteq B$, we have $\Pr[\text{$A$
    and $B$}] \geq \Pr[\text{$A$ and $B_T$}] = \Pr[A|B_T]\Pr[B_T] \geq
  \Pr[A|B_T]/2$,
  where the inequality follows from the argument of the first case.
  Since $S$ and $T$ are disjoint, the probability of selecting at
  least one edge from $S$ from among the potential
  $\binom{n}{2}-\card{T}$ edges not including $T$ is higher than the
  probability of selecting at least one edge from $S$ from among all
  potential $\binom{n}{2}$ edges.  We thus have $\Pr[\text{$A$ and
    $B$}] \geq \Pr[A|B_T]/2 \geq \Pr[A]/2$.  

  To complete the second case, we next lower-bound $\Pr[A]$.  Consider
  the process of sampling $G_D$ by toggling up to $k$ edges.  The
  number of possible graphs to sample from is thus
  $\sum_{i=0}^k \binom{n(n-1)/2}{i}$.  For $k \leq \binom{n}{2}/2$
  (which follows from $k \leq n/2$), the terms in the summation are
  strictly increasing.  So
  $\sum_{i=\ceil{k/2}}^k \binom{n(n-1)/2}{i} \geq (1/2) \sum_{i=0}^k
  \binom{n(n-1)/2}{i}$,
  i.e., with probability at least $1/2$ we toggle at least $k/2$
  edges.  If we choose (at least) $k/2$ random edges, the probability
  that none of them is in $S$ is at most
  $\left(1-\frac{\card{S}}{\binom{n}{2}}\right)^{k/2} \leq
  \left(1-\frac{(k/4)\card{S}}{\binom{n}{2}}\right)$
  following from the lemma assumption that the right-hand side is at
  least $1/2$.\footnote{The inequality can easily be proven by
    induction over the exponent: assuming the product so far satisfies
    $p \geq 1/2$, we have $p(1-x) \leq p-x/2$.}  Hence, when choosing
  at least $k/2$ edges, the probability of selecting at least one edge
  from $S$ is at least $(k/4)\card{S}/\binom{n}{2}$.  We conclude that
  $\Pr[A] \geq (1/2)(k/4)\card{S}/\binom{n}{2} \geq
  (1/8)k\card{S}/n^2$.
\end{proof}

\begin{lemma} \label{lem:connmissing} 
  There exists constant $c_2>0$ such that the following holds.
  Consider any graph $G = (V,E) \in allowed({\cal G}_{conn})$.
  Consider also any nonempty set $S \subseteq E$ of edges in the graph
  and smoothing value $k \leq n/16$.  Then with probability at most
  $c_2 k \card{S}/n^2$, the $k$-smoothed graph $G'$ removes an edge
  from~$S$.
\end{lemma}
\begin{proof}
  As in the proof of \lemref{connlb}, consider a graph $G_D$ sampled from
  $editdist(G,k)$, and let $B$ be the event that $G_D$ remains
  connected.  Let $A_R$ be the event that an edge from $S$ is
  deleted.  We want to upper bound $\Pr[A_R|B] = \Pr[\text{$A_R$ and
    $B$}] / \Pr[B] \leq \Pr[A_R]/\Pr[B]$.  As long as $\Pr[B] \geq 1/2$,
  we have $\Pr[A_R|B] \leq 2\Pr[A_R]$.  (Proving that $\Pr[B] \geq 1/2$ is
  virtually identical to the second case in proof of
  \lemref{connlb}.)

  We now upper-bound $\Pr[A_R]$. For $A_R$ to occur, an edge from $S$
  must be toggled.  So we can upper-bound $\Pr[A_R]$ using a union
  bound over the at most $k$ edge selections.  In particular, the
  $i$th edge selected belongs to $S$ with probability at most
  $2\card{S}/\binom{n}{2}$ (by the same argument as case~2 of
  \lemref{connlb}), giving $\Pr[A_R] \leq 2k\card{S}/\binom{n}{2}$. 
\end{proof}

Our next lemma is applicable when lower-bounding an algorithm's
performance on a dynamic graph.  It says essentially that
\lemref{connlb} is tight---it is not too likely to add any of the
helpful edges from $S$.

\begin{lemma} \label{lem:connub}
  There exists constant $c_3 > 0$ such that the following holds.
  Consider any graph $G=(V,E) \in allowed({\cal G}_{conn})$.  Consider
  also any set $S$ of edges and smoothing value $k \leq n/16$ such that
  $S\cap E = \emptyset$.  Then with probability at most
  $c_3k\card{S}/n^2$, the $k$-smoothed graph $G'$ of $G$ contains an
  edge from~$S$.
\end{lemma}
\begin{proof}
  This proof is identical to the proof of \lemref{connmissing}, with
  the event $A_R$ replaced by the event $A$ that at least one edge
  from $S$ is added.  (In either case, the event that at least one
  edge from $S$ is toggled by $G_D$.  Here, it is important that
  $S\cap E = \emptyset$ so that toggling is required to yield the edge
  in $G'$.) 
\end{proof}

\subsection{Pairing Network}
The second type we study is the {\em pairing network
  type}~\cite{cornejo:2012}.  This type is defined as
$({\cal G}_{pair}, C_{pair})$, where ${\cal G}_{pair}$ contains every
dynamic graph (defined with respect to our fixed node set $V$) in
which every individual graph is a (not necessarily complete) matching,
and $C_{pair}$ reliable communicates messages between pairs of nodes
connected in the given round. This network type is motivated by the
current peer-to-peer network technologies implemented in smart
devices. These low-level protocols usually depend on discovering
nearby nodes and initiating one-on-one local interaction.

\paragraph{Properties of Smoothed Pairing Networks.}
In the following, when discussing a matching $G$, we partition nodes
into one of two {\em types}: a node is {\em matched} if it is
connected to another node by an edge in $G$, and it is otherwise {\em
  unmatched}.  The following property concerns the probability that
smoothing affects (i.e., adds or deletes at least one adjacent edge) a
given node $u$ from a set $S$ of nodes of the same type.  It notes
that as the set $S$ containing $u$ grows, the {\em upper bound} on the
probability that $u$ is affected decreases. The key insight behind
this not necessarily intuitive statement is that this probability must
be the same for {\em all} nodes in $S$ (due to their symmetry in the
graph). Therefore, a given probability will generate more expected
changes as $S$ grows, and therefore, to keep the expected changes
below the $k$ threshold, this bound on this probability must decrease
as $S$ grows.

\begin{lemma} \label{lem:pairingsmooth} Consider any graph
  $G=(V,E) \in allowed({\cal G}_{pair})$ and constant $\delta > 1$.
  Let $S\subseteq V$ be a set of nodes in $G$ such that: (1) all nodes
  in $S$ are of the same type (matched or unmatched), and (2)
  $|S| \geq n/\delta$.  Consider any node $u\in S$ and smoothing
  factor $k < n/(2\cdot \delta)$.  Let $G'$ be the result of
  $k$-smoothing $G$.  The probability that $u$'s adjacency list is
  different in $G'$ as compared to $G$ is no more than
  $(2\cdot \delta\cdot k)/n$.
\end{lemma}
\begin{proof}
  For a given $w\in S$, let $p_w$ be the probably that $w$'s adjacency
  list is different in $G'$ as compared to $G$.  Assume for
  contradiction that $p_u > (2\cdot \delta\cdot k)/n$.  We first note
  that all nodes in $S$ are of the same type and therefore their
  probability of being changed must be symmetric. Formally, for every
  $u,v\in S$: $p_u = p_v$.

  For each $w\in S$, define $X_w$ to be the indicator random variable
  that is $1$ if $w$'s adjacency list changes in a particular
  selection of $G'$, and otherwise evaluates to $0$.  Let
  $Y=\sum_{w\in S}X_w$ be the total number of nodes in $S$ that end up
  changed in $G'$.  Leveraging the symmetry identified above, we can
  lower bound this expectation:

  \[  \mathbb{E}(Y) =  \mathbb{E}(\sum_{w\in S}X_w) = \sum_{w\in S} \mathbb{E}(X_w) = \sum_{w\in S} p_w = \sum_{w\in S} p_u > \frac{|S| \cdot 2\cdot  \delta\cdot k}{n} \geq 2k.\]

  \noindent Let $Z$ be the constant random variable that always
  evaluates to $2k$. We can interpret $Z$ as an upper bound on the
  total number of nodes that are affected by changes occurring in
  $G'$. (By the definition of $k$-smoothing there are at most $k$ edge
  changes, and each change affects two nodes.)  Because $Y$ counts the
  nodes affected by changes in $G'$, it follows that in all trials
  $Y \leq Z$.  The monotonicity property of expectation therefore
  provides that $\mathbb{E}(Y) \leq \mathbb{E}(Z) = 2k$.  We
  established above, however, that under our assumed $p_u$ bound that
  $\mathbb{E}(Y) > 2k$.  This contradicts our assumption that
  $p_u > (2\cdot \delta \cdot k)/n$.
\end{proof}

\section{Flooding}
\label{sec:flooding}

Here we consider the performance of a basic flooding process in a
connected dynamic network.  In more detail, we assume a single {\em
  source} node starts with a message.  In every round, every node that
knows the message broadcasts the message to its neighbors.  (Flooding
can be trivially implemented in a connected network type due to
reliable communication.)  We consider the flooding process {\em
  complete} in the first round that every node has the message.
Without smoothing, this problem clearly takes $\Omega(n)$ rounds in
large diameter static graphs, so a natural alternative is to state
bounds in terms of diameter.  Unfortunately, there exist dynamic
graphs (e.g., the spooling graph defined below) where the graph in
each round is constant diameter, but flooding still requires
$\Omega(n)$ rounds.

We show that this $\Omega(n)$ lower bound is somewhat fragile by
proving a polynomial improvement with any smoothing.  Specifically, we
show an upper bound of $O(n^{2/3}\log(n)/k^{1/3})$ rounds, with high
probability, with $k$-smoothing.  We also exhibit a nearly matching
lower bound by showing that the dynamic spooling graph requires
$\Omega(n^{2/3}/k^{1/3})$ rounds with constant probability.

\subsection{Lower Bound}

We build our lower bound around the dynamic {\em spooling graph},
defined as follows.  Label the nodes from $1$ to $n$, where node~$1$
is the source.  The spooling graph is a dynamic graph where in each
round $r$, the network is the $\min\set{r,n-1}$-spool graph.  We
define the {\em $i$-spool graph}, for $i\in [n-1]$ to be the graph
consisting of: a star on nodes $\set{1,\ldots,i}$ centered at $i$
called the {\em left spool}, a star on nodes $\set{i+1,\ldots,n}$
centered on $i+1$ called the {\em right spool}, and an edge between
the two centers $i$ and $i+1$.  We call $i+1$ the {\em head} node.

With node 1 as the source node, it is straightforward to see that, in
the absence of smoothing, flooding requires $n-1$ rounds to complete
on the spooling network. (Every node in the left spool has the message
but every node in the right spool does not.  In each round, the head
node receives the message then moves to the left spool.)  We
generalize this lower bound to smoothing.  The main idea is that in
order for every node to receive the message early, one of the early
heads must be adjacent to a smoothed edge.

\begin{theorem} \label{thm:floodinglower}
  Consider the flooding process on a $k$-smoothed $n$-vertex spooling
  graph, with $k\leq \sqrt{n}$ and sufficiently large $n$.  With
  probability at least $1/2$, the flooding process does not complete
  before the $\Omega(n^{2/3}/k^{1/3})$-th round.
\end{theorem}
\begin{proof}
  Consider the first $R-1=\delta n^{2/3}/k^{1/3}$ rounds of flooding
  on the spooling graph, where $\delta\leq 1$ is a constant to be
  tuned later.  Let $C$ be the set of nodes with ids from $1$ to $R$.
  These are the nodes that will play the role of head node during one
  of the $R-1$ rounds.  Our goal is to argue that, with constant
  probability, there is at least one node that does not receive the
  message during these rounds.

  There are three ways a node can learn the message during these $R-1$
  rounds: (1) it is the head and learns the message from the center of
  the left spool; (2) it is in the right spool and there is a smoothed
  edge between this node and a node with the message; (3) it is in the
  right spool and the current head node already has the message.  Case
  (1) is the easiest---the nodes $C$ receive the message this way.  We
  next bound the other two cases.  We say that we fail case~(2) if
  more than $R$ nodes receive the message by case~(2).  We say that we
  fail case~(3) if case (3) ever occurs in the first $R$ rounds.  We
  shall show that either case fails with probability at most $1/4$,
  and hence the probability of more than $2R$ nodes receiving the
  message is at most $1/2$.

  We first bound the probability of case (3) occurring.  This case can
  only happen if one of the nodes in $C$ received the message early
  due to a smoothed edge.  Consider each round $r$ during which we
  have not yet failed cases (2) or (3): i.e., at most $2R$ nodes in
  total have the message, but none of the nodes with ids between $r+1$
  and $R$ have it.  From \lemref{connub}, the probability of adding
  an edge between a node with the message and one of the relevant
  heads (those nodes in $C$) is
  $\leq c_3 k(2R)|C| / n^2 = O(kR^2/n^2) = O(\delta^2
  k^{1/3}/n^{2/3})$.
  Taking a union bound over all $R$ rounds, we get that the
  probability of failing case~(3) is at most
  $O(\delta^2 R k^{1/3}/n^{2/3}) = O(\delta^3)$.  For small enough
  $\delta$, this is at most $1/4$.

  We next argue that for small enough constant $\delta$, we fail case
  (2) with probability at most $1/4$.  To prove the claim, consider a
  round $r < R$ and suppose that we have not yet failed cases (2) or
  (3).  Thus at most $2R$ nodes have the message, and the probability
  that a specific node receives the message by case (2) in round $r$
  is at most $c_3 k (2R)/n^2=O(kR/n^2)$ by \lemref{connub}.  Thus, by
  linearity of expectation, the expected number of nodes receiving the
  message in round $r$ is $O(kR/n)$.  Summing over all $R-1$ rounds,
  the expected total number of nodes that learn the message this way
  is
  $O(kR^2/n) = O(k(\delta n^{2/3}/k^{1/3})^2/n) =
  O(\delta^2n^{1/3}k^{1/3}) = O(\delta R k^{2/3}/n^{1/3}) = O(\delta
  R)$
  for $k=O(\sqrt{n})$.  Thus, for $\delta$ small enough, the expected
  total number of nodes that receive the message by case~(2) is $R/4$.
  Applying Markov's inequality, the probability that more than $R$
  nodes receive the message is at most $1/4$.  We thus conclude that
  we fail case (2) with probability at most $1/4$.
\end{proof} 
 
\subsection{An $O(n^{2/3}\log{n} / k^{1/3})$ Upper Bound for General
  Networks}
 
We now show that flooding in {\em every} $k$-smoothed network will
complete in $O(n^{2/3}\log{n}/k^{1/3})$ time, with high
probability. When combined with the $\Omega(n^{2/3}/k^{1/3})$ lower
bound from above, this shows this analysis to be essentially tight for
this problem under smoothing.

{\bf Support Sequences.} The core idea is to show that every node in
every network is supported by a structure in the dynamic graph such
that if the message can be delivered to {\em anywhere} in this
structure in time, it will subsequently propagate to the target.  In
the spooling network, this structure for a given target node $u$
consists simply of the nodes that will become the head in the rounds
leading up to the relevant complexity deadline.  The {\em support
  sequence} object defined below generalizes a similar notion to all
graphs.  It provides, in some sense, a fat target for the smoothed
edges to hit in their quest to accelerate flooding.

\begin{definition}
  Fix two integers $t$ and $\ell$, $1 \leq \ell < t$, a dynamic graph
  $\mathcal H = G_1, \dots, G_t$ with $G_i=(V,E_i)$ for all $i$, and a
  node $u\in V$.  A {\em $(t,\ell)$-support sequence} for $u$ in $G$
  is a sequence $S_0,S_1,S_2,...,S_{\ell}$, such that the following
  properties hold: (1) For every $i\in [0,\ell]$: $S_i \subseteq V$.
  (2) $S_0 = \{u\}$. (3) For every $i\in [1,\ell]$:
  $S_{i-1} \subset S_i$ {\em and} $S_i \setminus S_{i-1} = \{v\}$, for
  some $v\in V$ such that $v$ is adjacent to at least one node of
  $S_{i-1}$ in $G_{t-i} $.
\end{definition}
Notice that the support structure is defined ``backwards'' with $S_0$
containing the target node $u$, and each subsequent step going one
round back in time.  We prove that every connected dynamic graph has
such a support structure, because the graph is connected in every
round.

\begin{lemma}
  Fix some dynamic graph $\mathcal H \in \mathcal G_{conn}$ on vertex
  set $V$, some node $u \in V$, and some rounds $t$ and $\ell$, where
  $1 \leq \ell < t$.  There exists a $(t,\ell)$-support sequence for
  $u$ in $\mathcal H$.
\label{lem:support}
\end{lemma}
\begin{proof}
  We can iteratively construct the desired support sequence.  The key
  observation to this procedure is the following: given any
  $S_i\subset V$ and static connected graph $G$ over $V$, there exists
  some $v\in V \setminus S_i$ that is connected to at least one node
  in $S_i$.  This follows because the absence of such an edge would
  imply that the subgraph induced by $S_i$ is a disconnected component
  in $G$, contradicting the assumption that it is connected.  With
  this property established, constructing the $(t,\ell)$-support
  sequence is straightforward: start with $S_0 = \{u\}$ and apply the
  above procedure again and again to the graph at times
  $t-1,t-2,...,t-\ell$ to define $S_1,S_2,...,S_{\ell}$ as required by
  the definition of a support sequence.
\end{proof}

The following key lemma shows that over every period of
$\Theta(n^{2/3}/k^{1/3})$ rounds of $k$-smoothed flooding, every node
has a constant probability of receiving the message.  Applying this
lemma over $\Theta(\log n)$ consecutive time intervals with a Chernoff
bound, we get our main theorem.

\begin{lemma}\label{lem:constantflood}
  There exists constant $\alpha \geq 3$ such that the following holds.
  Fix a dynamic graph $\mathcal H \in \mathcal G_{conn}$ on vertex set
  $V$, any node $u \in V$, and a consecutive interval of
  $\alpha n^{2/3}/k^{1/3}$ rounds.  For smoothing value $k \leq n/16$,
  node $u$ receives the flooded message in the $k$-smoothed version of
  $\mathcal H$ with probability at least~$1/2$.
\end{lemma}
\begin{proof}
  Let $t=\alpha n^{2/3}/k^{1/3}$, 
  and let $\ell=n^{2/3}/k^{1/3}$.
  Let ${\cal S} = S_0,S_1,\ldots,S_{\ell}$ be a $(t,\ell)$-support
  sequence for $u$ in $G$. (By Lemma~\ref{lem:support}, we know such a
  structure exists.) The key claim in this proof is the following:
 
  \smallskip
 
  {\em (*) If a node in $S_{\ell}$ receives the broadcast message by
    round $t-\ell$, then $u$ receives the broadcast message by round
    $t$ with probability at least $3/4$.}
 
  \smallskip
 
  To establish this claim we first introduce some further
  notation. Let $v_i$, for $i\in [\ell]$, be the single node in
  $S_i \setminus S_{i-1}$, and let $v_0 = u$.

  We will show by (reverse) induction the following invariant: for
  every $i \in [0, \ell]$, by the beginning of round $t-i$ there is
  some node $v_j$ with $j\leq i$ that has received the broadcast
  message.  For $i=0$, this implies that node $u = v_0$ has received
  the message and we are done.  The base case, where $i = \ell$,
  follows by the assumption that some node in $S_{\ell}$ receives the
  broadcast message prior to round $t-\ell$.

  Assume that the invariant holds for $i+1$; we will determine the
  probability that it holds for $i$.  That is, by the beginning of
  round $t-(i+1)$, there is some node $v_j$ for $j \leq i+1$ that has
  received the broadcast message.  If $j \leq i$, then the invariant
  is already true for $i$.  Assume, then, that $j=i+1$, i.e.,
  $v_{i+1}$ has received the message by the beginning of round
  $t-i-1$. By the definition of the $(t,\ell)$-support sequence, node
  $v_{i+1}$ is connected by an edge $e$ to some node in $S_i$, i.e.,
  to some node $v_{j'}$ for $j' \leq i$.  Thus if the specified edge
  $e$ is not removed by smoothing, then by the beginning of round
  $t-i$, there is some node $v_{j'}$ for $j' \leq i$ that has received
  the message.

  The probability that edge $e$ is removed by smoothing is at most
  $c_2 k / n^2$ (by Lemma~\ref{lem:connmissing}), for some constant
  $c_2$.  By taking a union bound over the $\ell = n^{2/3}/k^{1/3}$
  steps of the induction argument, we see that the claim fails with
  probability at most $c_2 k^{2/3} / n^{4/3} \leq 1/4$ for $k\leq n$
  and sufficiently large $n$.  This proves the claim.
 
  Now that we have established the value of getting a message to
  $S_\ell$ by round $t-\ell$ we are left to calculate the
  probability that this event occurs.  To do so, we first note that
  after $n^{2/3}/k^{1/3}$ rounds, there exists a set $T$ containing at
  least $n^{2/3}/k^{1/3}$ nodes that have the message. This follows
  because at least one new node must receive the message after every
  round (due to the assumption that
  $\mathcal H \in \mathcal G_{conn}$).
 
  There are therefore $|T||S_\ell|$ possible edges between $T$ and
  $S_\ell$, and if any of those edges is added via smoothing after
  round $n^{2/3}/k^{1/3}$ and before round
  $(\alpha-1)n^{2/3}/k^{1/3}$, then the target $u$ has a good chance
  of receiving the message.  In each such round, by
  Lemma~\ref{lem:connlb}, such an edge addition occurs with
  probability at least 
  $c_1 k |T||S_\ell|/n^2 \geq c_1 k^{1/3}/n^{2/3}$.

  Thus the probability that the message does not get to a node in
  $S_\ell$ during any of these $(\alpha-2)n^{2/3}/k^{1/3}$ rounds is
  $(1 - c_1 k^{1/3}/n^{2/3})^{(\alpha-2)n^{2/3}/k^{1/3}} \leq e^{- c_1
    (\alpha-2)}$.
  For a proper choice of $\alpha$, this shows that such an edge is
  added in at least one round between time $n^{2/3}/k^{1/3}$ and time
  $(\alpha - 1) n^{2/3}/k^{1/3}$ with probability at least $3/4$, and
  thus by time $t - \ell$ at least one node in $S_{\ell}$ has received
  the message.
 
  Putting the pieces together, we look at the probability of both
  events happening: the message reaches $S_{\ell}$ by round $t-\ell$,
  and the message propagates successfully through the support
  structure.  Summing the probabilities of error, we see that node $u$
  receives the message by time $t$ with probability at least $1/2$.
\end{proof}
 
\begin{theorem} \label{thm:floodingupper}
  For any dynamic graph $\mathcal H \in \mathcal G_{conn}$ and
  smoothing value $k\leq n/16$, flooding completes in
  $O(n^{2/3}\log{n}/k^{1/3})$ rounds on the $k$-smoothed version of
  $\mathcal H$ with high probability.
\end{theorem}
\begin{proof}
   Fix a non-source node $u$. We know via
   Lemma~\ref{lem:constantflood} that in each time interval of length
   $\Theta(n^{2/3}/k^{1/3})$, node $u$ receives the message with
   probability at least $1/2$.  Thus, over $\Theta(\log{n})$ such
   intervals, $u$ receives the message with high probability.  A union
   bound of the $n-1$ non-source nodes yields the final result.
\end{proof}

\section{Random Walks} \label{sec:random-walks}

As discussed in \secref{intro}, random walks in dynamic graphs exhibit
fundamentally different behavior from random walks in static graphs.
Most notably, in dynamic graphs there can be pairs of nodes whose
hitting time is exponential~\cite{avin:2008}, even though in static
(connected) graphs it is well-known that the maximum hitting time is
at most $O(n^3)$~\cite{lovasz:1996}.  This is true even under obvious
technical restrictions necessary to prevent infinite hitting times,
such as requiring the graph to be connected at all times and to have
self-loops at all
nodes.  

We show that this lower bound is extremely fragile.  A very simple
argument shows that a small perturbation ($1$-smoothing) is enough to
guarantee that in \emph{any} dynamic graph, all hitting times are at
most $O(n^3)$.  Larger perturbations ($k$-smoothing) lead to
$O(n^3 / k)$ hitting times.  We also prove a lower bound of
$\Omega(n^{5/2} / \sqrt{k})$, using an example which is in fact a
static graph (made dynamic by simply repeating it).

\subsection{Preliminaries}
We begin with some technical preliminaries.  In a static graph, a
random walk starting at $u \in V$ is a walk on $G$ where the next node
is chosen uniformly at random from the set of neighbors on the current
node (possibly including the current node itself if there is a
self-loop).  The \emph{hitting time} $H(u,v)$ for $u,v \in V$ is the
expected number of steps taken by a random walk starting at $u$ until
it hits $v$ for the first time.  Random walks are defined similarly in
a dynamic graph $\mathcal H = G_1, G_2, \dots$: at first the random
walk starts at $u$, and if at the beginning of time step $t$ it is at
a node $v_t$ then in step $t$ it moves to a neighbor of $v_t$ in $G_t$
chosen uniformly at random.  Hitting times are defined in the same way
as in the static case.

The definition of the hitting time in a smoothed dynamic graph is
intuitive but slightly subtle.  Given a dynamic graph $\mathcal H$ and
vertices $u,v$, the \emph{hitting time from $u$ to $v$ under
  $k$-smoothing}, denoted by $H_k(u,v)$, is the expected number of
steps taken by a random walk starting at $u$ until first reaching $v$
in the (random) $k$-smoothed version $\mathcal H'$ of $\mathcal H$
(either with respect to $\mathcal G_{conn}$ or with respect to the set
$\mathcal G_{all}$ of all dynamic graphs).  Note that this expectation
is now taken over two independent sources of randomness: the
randomness of the random walk, and also the randomness of the
smoothing (as defined in \secref{smoothing}).

\subsection{Upper Bounds}

We first prove that even a tiny amount amount of smoothing is
sufficient to guarantee polynomial hitting times even though without
smoothing there is an exponential lower bound.  Intuitively, this is
because if we add a random edge at every time point, there is always
some inverse polynomial probability of directly jumping to the target
node.  We also show that more smoothing decreases this bound linearly.

\begin{theorem} \label{thm:randomub} In any dynamic graph
  $\mathcal H$, for all vertices $u,v$ and value $k \leq n / 16$, the
  hitting time $H_k(u,v)$ under $k$-smoothing (with respect to
  $\mathcal G_{all}$) is at most $O(n^3 / k)$.  This is also true for
  smoothing with respect to $\mathcal G_{conn}$ if
  $\mathcal H \in \mathcal G_{conn}$.
\end{theorem}
\begin{proof}
  Consider some time step $t$, and suppose that the random walk is at
  some node $w$.  If $\{w,v\}$ is an edge in $G_t$ (the graph at time
  $t$), then the probability that it remains an edge under
  $k$-smoothing is at least $\Omega(1)$ (this is direct for smoothing
  with respect to $\mathcal G_{all}$, or from Lemma~\ref{lem:connlb}
  for smoothing with respect to $\mathcal G_{conn}$).  If $\{w,v\}$ is
  not an edge in $G_t$, then the probability that $\{w,v\}$ exists due
  to smoothing is at least $\Omega(k / n^2)$ (again, either directly
  or from Lemma~\ref{lem:connub}).  In either case, if this edge does
  exist, the probability that the random walk takes it is at least
  $1/n$.  So the probability that the next node in the walk is $v$ is
  at least $\Omega(k / n^3)$.  Thus at every time step the probability
  that the next node in the walk is $v$ is $\Omega(k / n^3)$, so the
  expected time until the walk hits $v$ is $O(n^3 / k)$.
\end{proof}

A particularly interesting example is the \emph{dynamic star}, which
was used by Avin et al.~\cite{avin:2008} to prove an exponential lower
bound.  The dynamic star consists of $n$ vertices
$\{0,1, \dots, n-1\}$, where the center of the start at time $t$ is $t
\mod (n-1)$ (note that node $n-1$ is never the center).  Every node
also has a self loop.  Avin et al.~\cite{avin:2008}
proved that the hitting time from node $n-2$ to node $n-1$ is at least
$2^{n-2}$.
It turns out that this lower bound is particularly fragile -- not only
does Theorem~\ref{thm:randomub} imply that the hitting time is
polynomial, it is actually a factor of $n$ better than the global
upper bound due to the small degrees at the leaves.

\begin{theorem} \label{thm:randomstar} $H_k(u,v)$ is at most
  $O(n^2 / k)$ in the dynamic star for all $k \leq n / 16$ and for all
  vertices $u,v$ (where smoothing is with respect to
  $\mathcal G_{conn}$).
\end{theorem}
\begin{proof}
  Suppose that the random walk is at some node $w$ at time $t$.  If
  $w$ is the center of the star at time $t$, then with constant
  probability $\{w,v\}$ is not removed by smoothing and with
  probability $\Omega(1/n)$ the random walk follows the edge
  $\{w,v\}$.  If $w$ is a leaf, then by Lemma~\ref{lem:connlb}, the
  probability that the edge $\{w,v\}$ exists in the smoothed graph is
  at least $\Omega(k/n^2)$.  On the other hand, it is straightforward
  to see that with constant probability no other edge incident on $w$
  is added (this is direct from Lemma~\ref{lem:connub} if $k = o(n)$,
  but for the dynamic star continues to hold up to $k = n$).  If this
  happens, then the degree of $w$ is $O(1)$ and thus the probability
  that the walk follows the randomly added edge $\{w,v\}$ is
  $\Omega(k / n^2)$.  Thus in every time step the probability that the
  walk moves to $v$ is $\Omega(k / n^2)$, and thus the hitting time is
  at most $O(n^2 / k)$.
\end{proof}

\subsection{Lower Bounds}

Since the dynamic star was the worst example for random walks in
dynamic graphs without smoothing, Theorem~\ref{thm:randomstar}
naturally leads to the question of whether the bound of $O(n^2 / k)$
holds for all dynamic graphs in $\mathcal G_{conn}$, or whether the
weaker bound of $O(n^3 / k)$ from Theorem~\ref{thm:randomub} is tight.
We show that under smoothing, the dynamic star is in fact \emph{not}
the worst case: a lower bound of $\Omega(n^{5/2} / \sqrt{k})$ holds
for the \emph{lollipop graph}. The lollipop is a famous example of
graph in which the hitting time is large: there are nodes $u$ and $v$
such that $H(u, v) = \Theta(n^3)$ (see, e.g., \cite{lovasz:1996}).
Here we will use it to prove a lower bound on the hitting time of
dynamic graphs under smoothing:

\begin{theorem} \label{thm:randomlb} There is a dynamic graph
  $\mathcal H \in \mathcal G_{conn}$ and nodes $u,v$ such that
  $H_k(u,v) \geq \Omega(n^{5/2} / (\sqrt{k} \ln n))$ for all
  $k \leq n / 16$ (where smoothing is with respect to
  $\mathcal G_{conn}$).
\end{theorem}
\begin{proof}
  In the lollipop graph $L_n = (V, E)$, the vertex set is partitioned
  into two pieces $V_1$ and $V_2$ with $|V_1| = |V_2| = n/2$.  The
  nodes in $V_1$ form a clique (i.e.~there is an edge between every
  two nodes in $V_1$), while the nodes in $V_2$ form a path (i.e.,
  there is a bijection $\pi : [n/2] \rightarrow V_2$ such that there
  is an edge between $\pi(i)$ and $\pi(i+1)$ for all
  $i \in [(n/2)-1]$).  There is also a single special node
  $v^* \in V_1$ which has an edge to the beginning of the $V_2$ path,
  i.e., there is also an edge $\{v^*, \pi(1)\}$.

  The dynamic graph $\mathcal H$ we will use is the dynamic lollipop:
  $G_i = L_n$ for all $i \geq 1$.  Let $u$ be an arbitrary node in
  $V_1$, and let $v = \pi(n/2)$ be the last node on the path.  We
  claim that $H_k(u,v) \geq \Omega(n^{5/2} / (\sqrt{k} \ln n))$.

  We will first define the notion of a \emph{phase}.  In an arbitrary
  walk on $\mathcal H$ under $k$-smoothing, a phase is a maximal time
  interval in which every node hit by the walk is in $V_2$ and all
  edges traversed are in $L_n$ (i.e., none of the randomly added
  smoothed edges are traversed).  The \emph{starting point} of a phase
  is the first vertex contained in the phase.

  Let
  $F = \{ w \in V_2 : \pi^{-1}(w) \geq (n/2 - c \sqrt{n / k} \ln n)\}$
  for some constant $c$ that we will determine later.  In other words,
  $F$ is the final interval of the path of length
  $t = c \sqrt{n/k} \ln n$.  We divide into two cases, depending on
  whether $v = \pi(n/2)$ is first hit in a phase that starts in $F$.
  We prove that in either case, the (conditional) hitting time is at
  least $\Omega(n^{5/2}/ (\sqrt{k} \log n))$.  Clearly this implies
  the theorem.

  \paragraph{Case 1:} Suppose that $v$ is first hit in a phase with
  starting point outside of $F$.  Consider such a phase (one with
  starting point outside of $F$).  Then by Lemma~\ref{lem:connlb} and
  the fact that the degree of every node in $V_2$ is at most $2$, the
  probability at each time step that the phase ends due to following a
  smoothed edge is at least
  $\frac{c_1 k (n-3)}{3 n^2} \geq \frac{ak}{n}$ for some constant
  $a > 0$.  Thus the probability that a phase lasts more than
  $\ell = \frac{bn}{k} \ln n$ steps is at most
  $(1-\frac{ak}{n})^{\ell} \leq e ^{-\ell ak / n} = e^{-ab \ln n} =
  n^{-ab}$.

  Now suppose that the phase last less than $\frac{bn}{k} \ln n$
  steps.  A standard Hoeffding bound implies that for all
  $\ell' \leq \ell = (bn/k) \ln n$, the probability that a walk of
  length $\ell$ is at a node more than $|F| = c\sqrt{n/k} \ln n$ away
  from the starting point of the phase is at most
  $e^{-2|F|^2 / \ell'} = e^{\frac{-2c^2 (n/k) \ln^2 n}{\ell'}}$.  Now
  a simple union bound over all $\ell' \leq \ell$ implies that the
  probability that the random walk hits $v$ during the phase
  (conditioned on the phase lasting at most $\ell$ steps) is at most
  $\ell \cdot e^{\frac{-2c^2 (n/k) \ln^2 n}{\ell}} = \frac{bn}{k} \ln
  n \cdot e^{\frac{-2c^2 \ln n}{b}} \leq b \ln n \cdot n^{-(2c^2-1) /
    b}$.

  To put these bounds together, let $A$ be the event that the random
  walk hits $v$ during the phase, and let $B$ be the event that the
  phase lasts more than $\ell$ steps.  Then
  \begin{align*}
    \Pr[A] &= \Pr[A|B] \Pr[B] + \Pr[A | \bar B] \Pr[\bar B] \leq \Pr[B] + \Pr[A | \bar B] \leq n^{-ab} + b \ln n \cdot n^{-(2c^2-1) / b}.
  \end{align*}

  For any constant $a$, if we set $b = 4/a$ and
  $c = \sqrt{2b + (1/2)}$ then we get that
  $\Pr[A] \leq O(n^{-4} \ln n) \leq O(n^{-3})$.  Hence the expected
  number of phases starting outside $F$ that occur before one of them
  hits $F$ is $\Omega(n^3)$, and thus the hitting time from $u$ to $v$
  (under $k$-smoothing) conditioned on $v$ first being hit in a phase
  beginning outside of $F$ is $\Omega(n^3)$.

  \paragraph{Case 2:} Suppose that $v$ is first hit in a phase with
  starting point in $F$.  We will show that the hitting time is large
  by showing that the expected time outside of such phases is large.
  We define two random variables.  Let $A_t$ be the number of steps
  between the end of phase $t-1$ and the beginning of phase $t$, and
  let $B_t$ be an indicator random variable for the event that the
  first $t-1$ phases all start outside of $F$ (where we set $B_1 = 1$
  by definition).  Then clearly the hitting time from $u$ to $v$,
  conditioned on $v$ being first hit in a phase with starting point in
  $F$, is $\E\left[\sum_{t=1}^{\infty}A_t B_t\right]$.  Since $A_t$
  and $B_t$ are independent, this is equal to
  $\sum_{t=1}^{\infty} \E[A_t] \E[B_t]$.

  A phase begins in one of two ways: either following a randomly added
  smoothed edge into $V_2$ (from either $V_1$ or $V_2$), or following
  the single edge in the lollipop from $V_1$ to $\pi(1)$.  If it
  begins by following a smoothed edge, then the starting point is
  uniformly distributed in $V_2$.  Since $\pi(1) \not\in F$, this
  clearly implies that
  $\E[B_t] \geq \left(1 - \frac{|F|}{|V_2|}\right)^{t-1} = \left(1 -
    \frac{2c \ln n}{\sqrt{nk}}\right)^{t-1}$.

  To analyze $\E[A_t]$, again note that phase $t-1$ ends by either
  following a smoothed edge or walking on the lollipop from $V_2$ to
  $V_1$.  Since the other endpoint of a smoothed edge is distributed
  uniformly among all nodes, the probability that phase $t-1$ ended by
  walking to $V_1$ is at least $1/2$.  So $\E[A_t]$ is at least $1/2$
  times the expectation of $A_t$ conditioned on phase $t-1$ ending in
  $V_1$.  This in turn is (essentially by definition) just $1/2$ times
  the expected time in $V_1$ before walking to $V_2$.  So consider a
  random walk that is at some node in $V_1$.  Clearly the hitting time
  to $\pi(1)$ without using smoothed edges is $\Omega(n^2)$ (we have a
  $1/n^2$ chance of walking to $v^*$ and then to $\pi(1)$).  The other
  way of starting a phase would be to follow a smoothed edge to $V_2$.
  By Lemma~\ref{lem:connub} the probability at each time step that the
  random walk is incident on a smoothed edge with other endpoint in
  $V_2$ is at most $O(k/n)$.  Since the degree of any node in $V_1$
  under $k$-smoothing is with high probability $\Omega(n)$, the
  probability that we follow a smoothed edge to $V_2$ if one exists is
  only $O(1/n)$, and hence the total probability of following a
  smoothed edge from $V_1$ to $V_2$ is at most $O(k/n^2)$.  Thus
  $\E[A_t] \geq \Omega(n^2 / k)$.

  So we get an overall bound on the hitting time (conditioned on $v$
  being first hit by a phase starting in $F$) of
  \begin{align*}
    \sum_{t=1}^{\infty} \E[A_t] \E[B_t] &\geq \Omega\left(\frac{n^2}{k}\right) \cdot \sum_{t=1}^{\infty} \left(1 - \frac{2c \ln n}{\sqrt{nk}}\right)^{t-1} \geq \Omega\left(\frac{n^{5/2}}{\sqrt{k} \ln n}\right).
  \end{align*} 

  Thus in both cases the hitting time is at least
  $\Omega\left(\frac{n^{5/2}}{\sqrt{k} \ln n}\right)$, completing the
  proof of the theorem.
\end{proof}

If we do not insist on the dynamic graph being connected at all times,
then in fact Theorem~\ref{thm:randomub} is tight via a very simple
example: a clique with a single disconnected node.

\begin{theorem} \label{thm:randomlball}
There is a dynamic graph $\mathcal H$ and vertices $u,v$ such that $H_k(u,v) \geq \Omega(n^3 / k)$ for all $k \leq n$ where smoothing is with respect to $\mathcal G_{all}$.
\end{theorem}
\begin{proof}
  Let each $G_i$ be the same graph: a clique on $n-1$ vertices and a
  single node $v$ of degree $0$.  Let $u$ be an arbitrary node in the
  clique, and consider a random walk starting from $u$.  At each time
  $t$, the probability that an edge exists from the current location
  of the walk to $v$ is $O(k / n^2)$.  If such an edge does exist,
  then the random walk follows it with probability $1 / (n-1)$.  Hence
  the total probability of moving to $v$ at any time is only
  $O(k / n^3)$, and thus the hitting time is $\Omega(n^3 / k)$.
\end{proof}

\section{Aggregation}
\label{sec:aggregation}
  
Here we consider the aggregation problem in the pairing dynamic
network type.  Notice, in our study of flooding and random walks we
were analyzing the behavior of a specific, well-known distributed
process.  In this section, by contrast, we consider the behavior of
arbitrary algorithms.  In particular, we will show the pessimistic
lower bound for the aggregation problem for $0$-smoothed pairing
graphs from~\cite{cornejo:2012}, holds (within constant factors), even
for relatively large amounts of smoothing.  This problem, therefore,
provides an example of where smoothing does not help much.
 
\paragraph{The Aggregation Problem.}
The aggregation problem, first defined in~\cite{cornejo:2012}, assumes
each node $u\in V$ begins with a unique token $\sigma[u]$.  The
execution proceeds for a fixed length determined by the length of the
dynamic graph.\footnote{This is another technical difference between
  the study of aggregation and the other problems considered in this
  paper.  For flooding and random walks, the dynamic graphs were
  considered to be of indefinite size. The goal was to analyze the
  process in question until it met some termination criteria.  For
  aggregation, however, the length of the dynamic graph matters as
  this is a problem that requires an algorithm to aggregate as much as
  it can in a fixed duration that can vary depending on the
  application scenario.  An alternative version of this problem can
  ask how long an algorithm takes to aggregate to a single node in an
  infinite length dynamic graph. This version of the problem, however,
  is less interesting, as the hardest case is the graph with no edges,
  which when combined with smoothing reduces to a standard random
  graph style analysis.}  At the end of the execution, each node $u$
{\em uploads} a set (potentially empty) $\gamma[u]$ containing
tokens. An {\em aggregation algorithm} must avoid both losses and
duplications (as would be required if these tokens were actually
aggregated in an accurate manner). Formally:
 
\begin{definition}
  An algorithm ${\cal A}$ is an {\em aggregation algorithm} if and
  only if at the end of
  every execution of ${\cal A}$ the following two properties hold:\\
(1) {\em No Loss:} $\bigcup_{u\in V} \gamma[u] = \bigcup_{u\in V} \{\sigma[u]\}$.
(2) {\em No Duplication:} $\forall u,v\in V, u\neq v: \gamma[u] \cap \gamma[v] = \emptyset$.
\end{definition}
 
To evaluate the performance of an aggregation algorithm we introduce
the notion of {\em aggregation factor.}  At at the end of an
execution, the aggregation factor of an algorithm is the number of
nodes that upload at least one token (i.e.,
$|\{ u\in V: \gamma[u] \neq \emptyset\} |$).  Because some networks
(e.g., a static cliques) are more suitable for small aggregation
factors than others (e.g., no edges in any round) we evaluate the
competitive ratio of an algorithm's aggregation factor as compared to
the offline optimal performance for the given network.

The worst possible performance, therefore, is $n$, which implies that
the algorithm uploaded from $n$ times as many nodes as the offline
optimal (note that $n$ is the maximum possible value for an
aggregation factor). This is only possible when the algorithm achieves
no aggregation and yet an offline algorithm could have aggregated all
tokens to a single node.  The best possible performance is a
competitive ratio of $1$, which occurs when the algorithm matches the
offline optimal performance.

\paragraph{Results Summary.}
In~\cite{cornejo:2012}, the authors prove that no aggregation
algorithm can guarantee better than a $\Omega(n)$ competitive ratio
with a constant probability or better.  In more detail:

\begin{theorem}[Adapted from~\cite{cornejo:2012}]
  For every aggregation algorithm ${\cal A}$, there exists a pairing
  graph ${\cal H}$ such that with probability at least $1/2$:
  ${\cal A}$'s aggregation factor is $\Omega(n)$ times worse than the
  offline optimal aggregation factor in ${\cal H}$.
\label{thm:agg}
\end{theorem}

Our goal in the remainder of this section is to prove that this strong
lower bound persists even after a significant amount of smoothing
(i.e., $k=O(n/\log^2{n})$).  We formalize this result below (note that
the cited probability is with respect to the random bits of both the
algorithm and the smoothing process):

\begin{theorem}
  For every aggregation algorithm ${\cal A}$ and smoothing factor
  $k \leq n/(32\cdot \log^2{n})$, there exists a pairing graph
  ${\cal H}$ such that with probability at least $1/2$: ${\cal A}$'s
  aggregation factor is $\Omega(n)$ times worse than the offline
  optimal aggregation factor in a $k$-smoothed version of ${\cal H}$
  (with respect to $\mathcal G_{pair}$).
\label{thm:agg:simplified}
\end{theorem}

\subsection{Lower Bound}
 
Here we prove that for any smoothing factor $k \leq (cn)/\log^2{n}$
(for some positive constant fraction $c$ we fix in the analysis),
$k$-smoothing does not help aggregation by more than a constant factor
as compared to $0$-smoothing.  To do so, we begin by describing a
probabilistic process for generating a hard pairing graph.  We will
later show that the graph produced by this process is likely to be
hard for a given randomized algorithm.  To prove our main theorem, we
will conclude by applying the probabilistic method to show this result
implies the existence of a hard graph for each algorithm.

\paragraph{The $\alpha$-Stable Pairing Graph Process.}
We define a specific process for generating a pairing graph (i.e., a
graph in $allowed({\cal G}_{pair})$).  The process is parameterized by
some constant integer $\alpha \geq 1$.
In the following, assume the network size $n=2\ell$ for some integer
$\ell\geq 1$ that is also a power of $2$.\footnote{We can deal with
  odd $n$ and/or $\ell$ not a power of $2$ by suffering only a
  constant factor cost to our final performance.  For simplicity of
  presentation, we maintain these assumptions for now.}  For the
purposes of this construction, we label the $2\ell$ nodes in the
network as $a_1, b_1, a_2,b_2,...,a_{\ell},b_{\ell}$.
For the first $\alpha$ rounds, our process generates graphs with the
edge set: $\{(a_i, b_i): 1 \leq i \leq \ell\}$.  After these rounds,
the process generates $\ell$ bits, $q_1,q_2,...,q_{\ell}$, with
uniform randomness.  It then defines a set $S$ of {\em selected nodes}
by adding to $S$ the node $a_i$ for every $i$ such that $q_i =0$, and
adding $b_i$ for every $i$ such that $q_i = 1$.  That is, for each of
our $(a_i, b_i)$ pairs, the process randomly flips a coin to select a
single element from the pair to add to $S$.
 
For all graphs that follow, the nodes {\em not in} $S$ will be
isolated in the graph (i.e., not be matched).  We turn our attention
to how the process adds edges between the nodes that are in $S$.  To
do so, it divides the graphs that follow into {\em phases}, each
consisting of $\alpha$ consecutive rounds of the same graph.  In the
first phase, this graph is the one that results when the process pairs
up the nodes in $S$ by adding an edge between each such pair (these
are the only edges).  In the second phase, the process defines a set
$S_2$ that contains exactly one node from each of the pairs from the
first phase. It then pairs up the nodes in $S_2$ with edges as
before. It also pairs up all nodes in $S \setminus S_2$ arbitrarily.
Every graph in the second phase includes only these edges.  In the
third phase, the process defines a set $S_3$ containing exactly one
node from each of the $S_2$ pairs from the previous pairs. It then
once again pairs up the remaining nodes in $S$ arbitrarily.  The
process repeats this procedure until phase $t = \log_2{|S|}$ at which
point only a single node is in $S_t$, and we are done.

The total length of this dynamic graph is $\alpha(\log_2{(|S|)} + 1)$.
It is easy to verify that it satisfies the definition of the pairing
dynamic network type.
 
\paragraph{Performance of the Offline Optimal Aggregation Algorithm.}
We now show that the even with lots of smoothing, a graph generated by
the stable pairing graph process, parameterized with a sufficiently
large $\alpha$, yields a good optimal solution (i.e., an aggregation
factor of $1$).
 
\begin{lemma}
  For any $k \leq n/32$, and any pairing graph ${\cal H}$ that might
  be generated by the $(\log{n})$-stable pairing graph process, with
  high probability in $n$: the offline optimal aggregation algorithm
  achieves an aggregation factor of $1$ in a $k$-smoothed version of
  ${\cal H}$.
\label{lem:aggopt}
\end{lemma} 
\begin{proof}
  In the following, let $\alpha=\log{n}$ be the stability parameter
  provided the graph process.  We will describe an offline algorithm
  that guarantees for every possible dynamic graph produced by the
  $\alpha$-stable pairing graph process, that with high probability
  (over the smoothing choices) the algorithm aggregates all values to
  a single node.
  It follows, of course, that the offline optimal algorithm achieves
  this same optimal factor.

  To begin our argument, consider some dynamic graph generated by the
  $\alpha$-stable pairing process.  We divide its graphs into {\em
    phases} of $\alpha$ consecutive static graphs each.  Recall from
  our definition of the graph process that during each phase, some
  nodes are paired and some nodes are isolated.  It also follows from
  our definition of this process that the graphs within a given phase
  are all the same.

  Fix some phase $i$ and some edge $(u,v)$ in that phase's graph.  Our
  key observation is that this edge is included in at least one of the
  smoothed graphs during this phase (i.e., there is at least one round
  where the smoothing does not remove this edge), with high
  probability in $n$. To prove this observation, we first focus on a
  single round in this phase and show there is at least a constant
  probability that $(u,v)$ is left alone in this round.  To prove this
  result we apply \lemref{pairingsmooth} to $S=M$ (where $M$ is the
  set of matched nodes in this round), $\delta=2$, and node $u$.  (We
  obtain our $\delta$ value by noting that, by definition, the set of
  matched nodes in every round of every graph produced by our
  adversary includes at least half the nodes.)  The probability that
  $(u,v)$ is removed in a our round is less than or equal to the
  probability $p_u$ that $u$ is affected in the round.
  \lemref{pairingsmooth} tells us that $p_u \leq (4k)/n$.  Given our
  assumption that $k \leq n/32$, it follows that
  $p_u \leq 4/32 = 1/8$.  Therefore, the probability that $(u,v)$ is
  removed by smoothing in {\em every} round of phase $i$, is no more
  than:

  \[ p_u^{\alpha} = (1/8)^{\log{n}} = (1/2^3)^{\log{n}} = 1/n^3.\]

  To complete the proof, we apply two union bounds to obtain this
  property for all matched edges in all phases, while avoiding
  dependencies.  First, within a given phase, a union bound provides
  that with probability $1/n^2$, all matched edges (of which there are
  less than $n$) in the original graph are preserved at least once.
  Another union bound for the total number of phases (which is also
  less than $n$), provides that this holds for very phase.  It is then
  simple to verify from the definition of the smoothing process that
  it is possible to aggregate every value to the single node that
  remains in the $S_t$ (i.e., the root of the tree induced by the
  phases greater than $1$).
\end{proof}
 
\paragraph{Performance of an Arbitrary Distributed Aggregation
  Algorithm.}
 We now fix an arbitrary distributed aggregation algorithm and
 demonstrate that it cannot guarantee (with good probability) to
 achieve a non-trivial competitive ratio in all pairing graphs.  In
 particular, we will show it has a constant probability of performing
 poorly in a graph generated by our above process.
 
\begin{lemma}
  Fix an online aggregation algorithm ${\cal A}$ and smoothing factor
  $k \leq n/ (32\cdot \log^2{n})$.  Consider a $k$-smoothed version of
  a graph ${\cal H}$ generated by the $(\log{n})$-stable pairing graph
  process.  With probability greater than $1/2$ (over the smoothing,
  adversary, and algorithm's independent random choices): ${\cal A}$
  has an aggregation factor in $\Omega(n)$ when executed in this
  graph.
 \label{lem:aggalg}
\end{lemma}
\begin{proof}
  Consider the pairing graph generated by the stable pairing graph
  process, before the smoothing is applied.
  Let $S$ be the set of {\em selected} nodes (see the process
  definition) and $\bar S = V\setminus S$ be the non-selected nodes.
  By definition, $|\bar S| = |S| = n/2$.

  Let $Y$ be the number of times that nodes in $\bar S$ are affected
  by smoothing in the $s = O(\log^2{n})$ individual graphs generated
  by the $(\log{n})$-stable pairing graph process.  To calculate this
  expectation, let $X_{i,r}$ be the indicator random variable that is
  $1$ if the $i^{th}$ node in $\bar S$ (by some arbitrary but fixed
  ordering) is affected by smoothing in the $r^{th}$ graph generated
  by the process, and otherwise $0$.  It follows that
  $Y = \sum X_{i,r}$, for all $i\in [|\bar S|]$ and $r\in [s]$, and
  therefore
  $\mathbb{E}(Y) = \sum \mathbb{E}(X_{i,r}) = \sum \Pr(X_{i,r} =1)$.
 
  We can upper bound $\Pr(X_{i,r} =1)$, for any particular $i$ and
  $r$, by applying \lemref{pairingsmooth} to the set $\bar S$ of
  unmatched nodes and $\delta=2$.  It follows that:
  \[Pr(X_{i,r} =1) \leq (2\cdot \delta\cdot k)/n = (4k)/n \leq 4/(32
  \log^2{n}) = 1/(8\log^2{n}). \]%
  \noindent Because we are summing over
  $|\bar S|\cdot s \leq (n/2)\log^2{n}$ values, we obtain the bound
  $\mathbb{E}(Y) \leq n/16$.  We now apply Markov's inequality to
  establish that $Y \geq n/4$ with probability no more than $1/4$.
  Let $U$ denote the set of nodes in $\bar S$ that remain
  undistributed by smoothing.  Notice, we just proved that probability
  at least $3/4$, $|U| \geq n/4$.
 
  We now want to consider what happens to nodes in $U$.  Fix some
  $x\in U$. Let $y\in S$ be the node connected to $x$ throughout phase
  $1$ in the graph generated by the graph process. Because $x\in U$,
  it follows that this edge is undisturbed throughout phase $1$.  It
  follows that $x$ and $y$ have no opportunity to learn new tokens or
  to pass on their existing tokens outside the pair during phase $1$.
  A key property proved in~\cite{cornejo:2012} is that at the end of
  phase $1$ we can assign {\em owners} to $x$ and $y$'s tokens among
  $x$ and $y$.  To do so, consider what would happen if we extend this
  graph such that going forward $x$ and $y$ are isolated.  It must be
  the case that $\gamma[x] \cup \gamma[y] = \{\sigma[x],\sigma[y]\}$
  and $\gamma[x] \cap \gamma[y] = \emptyset$. If these properties did
  not hold in this extension, then no duplication and/or no loss would
  be violated in this extension, and the aggregation problem requires
  these properties to {\em always} hold.  It follows that at least one
  of these two nodes has a non-empty $\gamma$ set at the end of this
  isolation extension.  Because the graph process selected which node
  went to $S$ from among this pair with uniform and independent
  randomness, the probability that the node not chosen for $S$ ends up
  owning at least one token (by our above definition of ownership), is
  at least $1/2$.  Put another way, for each node in $U$, with
  independent probability at least $1/2$, that node will end up
  outputting at least one token at the end of the execution.  We
  expect that at least half the nodes in $U$ will output tokens.
  Another application of Markov's tells us that the probability that
  less than a smaller constant fraction of these nodes end up
  uploading is less than $1/4$.
 
  Pulling together the pieces, we first proved that the probability
  that $|U| < n/4$ is at most $1/4$.  We then proved that the
  probability that less than $|U|/j$ nodes end up uploading, for some
  constant $j$, is also at most $1/4$.  A union bound provides that
  the probability at least one of these bad events occurs is strictly
  less than $1/2$.  In the case that neither bad event occurs (a case
  that holds with probability strictly more than $1/2$), we end up
  with an aggregation factor in $(1/j)(n/4) \in \Omega(n)$---as
  required by the lemma statement.
\end{proof}
 
A final union bound combines the results from Lemmas~\ref{lem:aggopt}
and~\ref{lem:aggalg} to get our final corollary.  Applying the
probabilistic method to the corollary yields the main
theorem---\thmref{agg:simplified}.
   
\begin{corollary}
  Fix an aggregation algorithm ${\cal A}$ and smoothing factor
  $k \leq n/ (32\cdot \log^2{n})$.  There is a method for
  probabilistically constructing a pairing graph ${\cal H}$, such that
  with probability greater than $1/2$ (over the smoothing, adversary,
  and algorithm's independent random choices): ${\cal A}$'s
  aggregation factor in a $k$-smoothed version of ${\cal H}$ is
  $\Omega(n)$ times larger than the offline optimal factor for this
  graph.
  \label{corr:agg}
\end{corollary}

\bibliographystyle{plain}
\bibliography{smoothing}

\end{document}